\documentclass[10pt]{article}
\usepackage{}

\usepackage{pgf}
\usepackage[square, comma, sort&compress, numbers]{natbib}
\usepackage[english]{babel}
\usepackage[utf8]{inputenc}
\usepackage{datetime}
\usepackage{geometry}
\usepackage{array}
\usepackage{xspace}

\usepackage{amsmath,amsfonts,amssymb,amsopn,amsthm,amscd}
\theoremstyle{plain}
\usepackage{graphicx} 
\usepackage{epstopdf}
\usepackage{enumerate}
\allowdisplaybreaks[4]

\def\Box{\vcenter{\vbox{\hrule\hbox{\vrule
     \vbox to 8.8pt{\hbox to 10pt{}\vfill}\vrule}\hrule}}}

\usepackage{indentfirst}
\newtheorem{thm}{Theorem}[section]
\newtheorem{lem}[thm]{Lemma}

\newtheorem{cor}[thm]{Corollary}

\newtheorem{prop}[thm]{Proposition}

\newtheorem{remark}{Remark}

\begin{document}

\begin{center}

{\large  \bf Minimal Linear Codes Constructed from partial spreads}

\vskip 0.8cm
{\small Wei Lu$^1$, Xia Wu\footnote{Supported by NSFC (Nos. 11971102, 12171241), the Fundamental Research Funds for the Central Universities.

  }}$^*$, Xiwang Cao$^2$, Gaojun Luo$^3$ , Xupu Qin$^1$\\

{\small $^1$School of Mathematics, Southeast University, Nanjing
210096, China}\\
{\small $^2$Department of Math, Nanjing University of Aeronautics and Astronautics, Nanjing 211100, China}\\
{\small $^3$School of Physical and Mathematical Sciences, Nanyang Technological University (NTU), Singapore
637371}\\
{\small E-mail:
 luwei1010@seu.edu.cn, wuxia80@seu.edu.cn, xwcao@nuaa.edu.cn, gaojun.luo@ntu.edu.sg,  2779398533@qq.com}\\
{\small $^*$Corresponding author. (Email: wuxia80@seu.edu.cn)}
\vskip 0.8cm
\end{center}

{\bf Abstract:}
Partial spread  is   important in finite geometry and  can  be used to construct linear codes.
 From the results in (Designs, Codes and Cryptography 90:1-15, 2022) by Xia Li, Qin Yue and Deng Tang, we know that if  the number of the elements in a partial spread is  ``big enough", then the corresponding linear code is minimal. They used the sufficient condition in (IEEE Trans. Inf. Theory 44(5): 2010-2017, 1998) to prove the minimality of such linear codes. In this paper, we use the geometric approach to
study the minimality of linear codes constructed from partial spreads in all cases.


\section{\bf Introduction}

 Let $q$ be a prime power and $\mathbb{F}_q$ the finite field with $q$ elements. Let $n$ be a positive integer and $\mathbb{F}_q^n$  the vector space with dimension $n$ over $\mathbb{F}_q$. In this paper, all vector spaces are over $\mathbb{F}_q$ and  all vectors are row vectors.  For a vector $\mathbf{v}=(v_1, \dots, v_n)\in\mathbb{F}_q^n$, let Suppt$(\mathbf{v})$ $:= \{1 \leq i\leq n : v_i\neq 0\}$ be the support of $\mathbf{v}$. The \emph{Hamming weight} of vector of $\mathbf{v}$ is wt$(\mathbf{v})$:=\# $\rm{Suppt}(\mathbf{v})$.  For any two vectors $\mathbf{u}, \mathbf{v}\in \mathbb{F}_q^n$, if $\rm{Suppt}(\mathbf{u})\subseteq \rm{Suppt}(\mathbf{v})$, we say that $\mathbf{v}$ covers $\mathbf{u}$ (or $\mathbf{u}$ is covered by $\mathbf{v}$) and write $\mathbf{u}\preceq\mathbf{v}$. Clearly, $a\mathbf{v}\preceq \mathbf{v}$ for all $a\in \mathbb{F}_q$.

 An $[n,m]_q$ linear code $\mathcal{C}$ over $\mathbb{F}_q$ is a $m$-dimensional subspace of $\mathbb{F}_q^n$. Vectors in $\mathcal{C}$ are called codewords. A codeword $\mathbf{c}$ in a linear code $\mathcal{C}$ is called \emph{minimal} if $\mathbf{c}$ covers only the codewords $a\mathbf{c}$ for all $a\in \mathbb{F}_q$, but no other codewords in $\mathcal{C}$. That is to say, if a codeword $\mathbf{c}$  is minimal in  $\mathcal{C}$, then for any codeword $\mathbf{b}$ in $\mathcal{C}$, $\mathbf{b}\preceq \mathbf{c}$ implies that $\mathbf{b}=a\mathbf{c}$ for some $a\in \mathbb{F}_q$.
  For an arbitrary linear code $\mathcal{C}$, it is  hard  to determine the set of its minimal codewords \cite{BMT1978, BN1990}.

If every codeword in  $\mathcal{C}$ is minimal, then $\mathcal{C}$ is said to be a \emph{minimal linear code}. Minimal linear codes have interesting applications in secret sharing \cite{CDY2005, CCP2014, DY2003, M1995, YD2006} and secure two-party computation \cite{ABCH1995, CMP2013}, and could be decoded with a minimum distance decoding method \cite{AB1998}. Searching for minimal linear codes has been an interesting research topic in coding theory and cryptography.

Up to now, there are two approaches to study minimal linear codes. One is  algebraic method and the other  is geometric method. The algebraic method is based on the Hamming weights of the codewords. 
In \cite{AB1998}, Ashikhmin and Barg gave a sufficient condition  on the minimum and maximum nonzero Hamming weights for a linear code to be minimal.
\begin{lem}\label{Ashikhmin-Barg}{$($\rm {Ashikhmin-Barg} \cite{AB1998}$)$}
A linear code $\mathcal{C}$ over $\mathbb{F}_q$ is minimal if
$$\frac{w_{\rm min}}{w_{\rm max}}>\frac{q-1}{q},$$
where $w_{\rm min}$ and $w_{\rm max}$ denote the minimum and maximum nonzero Hamming weights in the code $\mathcal{C}$, respectively.
\end{lem}
Cohen et al. \cite{CMP2013} provided an example to show that the condition $\frac{w_{\rm min}}{w_{\rm max}}>\frac{q-1}{q}$  is not necessary for a linear code to be minimal. Ding, Heng and Zhou \cite{DHZ2018, HDZ2018}  generalized this  sufficient condition and derived a sufficient and necessary condition on all Hamming weights for a given linear code to be
minimal.

When using the algebraic method to prove the minimality of a given linear code, one needs to know all the Hamming weights in the code, which is very difficult in general. Even all the
Hamming weights are known, it is also hard to use the algebraic method   to prove the minimality. In this paper, we will show that there exist two linear codes with  the same weight distribution, one is minimal but the other is not.

Recently, minimal linear codes were characterized by geometric approaches in \cite{ABN2019,LW2021,TQLZ2019}. In \cite{ABN2019,TQLZ2019}, the authors used the cutting blocking sets to study the minimal linear codes. In \cite{LW2021}, the authors used the basis of linear space to study the minimal linear codes. Based on these results, it is  easier to construct minimal linear codes or to prove the minimality of some linear codes , see \cite{ABNR2019,BB2019,BB2021,BCMP2021,MB2021,HN}.

There are two basic problems in minimal linear codes. One  is the existence of minimal linear codes with given parameters, see \cite{ABN2019,LW2021,TQLZ2019}. The other is to discuss the minimality of some interesting  codes, see \cite{LYT2022,XQC2020,XQ2019}. In this paper, we will discuss the minimality of linear codes constructed by partial spread.

Partial spread (see section \ref{section Preliminaries}) is   important in finite geometry \cite{KM2017} and can be used to construct bent functions \cite{AM2022}.   Partial spread can also be used to construct linear codes.
Let $s$ be the number of the elements in a  partial spread.  From \cite[Theorem 10]{LYT2022}, we know that if  $s$ is ``big enough", (precisely $s\geq q+1$), then the corresponding linear code is minimal. In \cite[Theorem 10]{LYT2022}, they used the sufficient condition in Lemma \ref{Ashikhmin-Barg} to prove the minimality of such linear codes.

In this paper, we will use the geometric approach to
consider the minimality of linear codes constructed from partial spreads for all $s$. We get the following three results:
  (1) when $s$ is  ``big enough", precisely, when  $s\geq q+1$, for any partial spread,  the corresponding linear code is minimal;
    (2) when $s$ is  ``small enough", precisely, when  $2\leq s\leq 3\leq q$,  for any partial spread, the corresponding linear code is not minimal;
    (3) when $3<s\leq q$, for some  partial spreads,  the corresponding linear codes are minimal, while for some other  partial spreads,  the corresponding linear codes are not  minimal.

The rest of this paper is organized as follows.

 \section{\bf Preliminaries}\label{section Preliminaries}

 \subsection{Partial spread}

Throughout this paper, let $k$ be a positive integer and $m=2k$. A \emph{partial spread}  $\mathcal{S}$ of $\mathbb{F}_q^{m}$  is a set of $k$-dimensional subspaces of  $\mathbb{F}_q^{m},$  which pairwise intersect trivially. It is easy to see that $\#\mathcal{S}\leq q^k+1.$  If $\#\mathcal{S}=q^k+1,$ hence every nonzero element of $\mathbb{F}_q^{m}$ is in exactly one of those subspaces, then $\mathcal{S}$ is called a (complete) spread.

Let $2\leq s\leq q^k+1$ be a positive integer and
\begin{equation}\label{eq1}
\Omega=\{E_i\leq \mathbb{F}_q^m:\ {\rm{dim}}E_i=k,\ E_i\cap E_j=\{\mathbf{0}\},1\leq i\neq j\leq s\}.
\end{equation}
 Then $\Omega$ is a partial spread of $\mathbb{F}_q^{m}$. It is easy to see that for any $1\leq i\neq j\leq s,$

  \begin{equation}\label{sum}
  E_{i}+E_{j}=\mathbb{F}_q^m.
  \end{equation}

 \subsection{Euclidean inner product}
 Let $m$ be a positive integer. For  vectors $\mathbf{x}=(x_1,x_2,...,x_m)$, $\mathbf{y}=(y_1,y_2,...,y_m)\in\ \mathbb{F}_q^m$, their \emph{Euclidean inner product} is:
$$<\mathbf{x},\mathbf{y}>:=\mathbf{x}\mathbf{y}^T=\sum_{i=1}^m{x}_i{y}_i.$$
For any $S\subseteq\mathbb{F}_q^m $, we define
$${\rm{Span}}(S):=\{\sum_{i=1}^r{\lambda}_i{\mathbf{s}}_i\  |\  r\in \mathbb{N}, {\mathbf{s}}_i\in S, \lambda_i\in\mathbb{F}_q\},$$
$$S^\perp:=\{\mathbf{v}\in \mathbb{F}_q^m\ |\ \mathbf{vs}^T=0,\  {\rm{for\  any}}\  \mathbf{s}\in S\}.$$
Then Span$(S)$ and $S^\perp$ are vector spaces over $\mathbb{F}_q$ and
\begin{equation}\label{dim sum}
{\rm{dim}({\rm{Span}}}(S))+{\rm{dim}}(S^\perp)=m.
\end{equation}


 \subsection{Minimal linear codes}

All linear codes can be constructed by the following way. Let $m\leq n$ be two positive integers. Let $D:=\{\mathbf{d}_1,...,\mathbf{d}_n\}$  be a multiset, where $\mathbf{d}_1,...,\mathbf{d}_n\in \mathbb{F}_q^m$.   Let  $r(D)$ be the rank of $D$ (it equals  the dimension of the vector space Span$(D)$ over $\mathbb{F}_q$). Let

 $$\mathcal{C}=\mathcal{C}(D)=\{{\mathbf{c}\mathbf{(x)}}=\mathbf{c}(\mathbf{x};D)=(\mathbf{xd}_1^{T},...,\mathbf{xd}_n^{T}), \mathbf{x}\in \mathbb{F}_q^m\}.$$
Then $\mathcal{C}(D)$ is an $[n$, $r(D)]_q$ linear code. We always study the minimality of $\mathcal{C}(D)$ by considering some appropriate  multisets $D$.

To present the  sufficient and necessary condition for minimal linear codes in \cite{LW2021}, some concepts are needed.

For any $\mathbf{y}\in \mathbb{F}_q^m$, we define
$$H(\mathbf{y}):=\mathbf{y}^\perp=\{\mathbf{x}\in \mathbb{F}_q^m\mid\mathbf{xy}^{T}=0\},$$
$$H(\mathbf{y},D):=D\cap H(\mathbf{y})=\{\mathbf{x}\in D\mid\mathbf{xy}^{T}=0\},$$
$$V(\mathbf{y},D):={\rm{Span}}(H(\mathbf{y},D)).$$
It is obvious that $H(\mathbf{y},D)\subseteq V(\mathbf{y},D)\subseteq H(\mathbf{y})$.

\begin{prop}\label{21}\cite{LW2021}
For any $\mathbf{x}, \mathbf{y}\in \mathbb{F}_q^m,\ \mathbf{c(x)}\preceq \mathbf{c(y)}$ if and only if $H(\mathbf{y},D)\subseteq H(\mathbf{x},D)$.
\end{prop}

Let $\mathbf{y}\in \mathbb{F}_q^m\backslash \{\mathbf{0}\}$. The following lemma   gives a sufficient and necessary condition for the codeword $\mathbf{c(y)}\in \mathcal{C}(D)$ to be minimal.

\begin{lem}\label{sn}\cite[Theorem 3.1]{LW2021}
Let $\mathbf{y}\in \mathbb{F}_q^m\backslash \{\mathbf{0}\}$. Then the following three conditions are equivalent:\\
$(1)$ $\mathbf{c(y)}$ is minimal in $\mathcal{C}(D)$;\\
$(2)$ \rm{dim}$V(\mathbf{y},D)=m-1$;\\
$(3)$ $V(\mathbf{y},D)=H(\mathbf{y})$.
\end{lem}

The following lemma gives a sufficient and necessary condition for linear codes over $\mathbb{F}_q$ to be minimal.
\begin{lem}\label{lem23}\cite[Theorem 3.2]{LW2021}
The following three conditions are equivalent:\\
$(1)$ $\mathcal{C}(D)$ is minimal;\\
$(2)$ for any $\mathbf{y}\in \mathbb{F}_q^m\backslash \{\mathbf{0}\}$,  $\rm{dim}$$V(\mathbf{y},D)=m-1$;\\
$(3)$ for any $\mathbf{y}\in \mathbb{F}_q^m\backslash \{\mathbf{0}\}$, $V(\mathbf{y},D)=H(\mathbf{y})$.

\end{lem}
By the following lemma, we can get infinity many minimal linear codes from any known minimal linear codes.
\begin{lem}\label{lem24}\cite[Proposition 4.1.]{LW2021}
Let $D_1\subseteq D_2$ be two multisets with elements in $ \mathbb{F}_q^m$ and $r(D_1)=r(D_2)=m$. If $\mathcal{C}(D_1)$ is minimal, then $\mathcal{C}(D_2)$ is minimal.
\end{lem}
The following corollary is trivial.
\begin{cor}\label{cor21}
Let $D_1\subseteq D_2$ be two multisets with elements in $ \mathbb{F}_q^m$ and $r(D_1)=r(D_2)=m$. If $\mathcal{C}(D_2)$ is not minimal, then
$\mathcal{C}(D_1)$ is not minimal.
\end{cor}

In the following section, we will use the above lemmas to consider the minimality of linear codes constructed from partial spreads.

\section{\bf The minimality of  linear codes constructed from partial spreads}\label{section main theorem}

In this section, we consider the linear codes constructed from partial spreads and discuss the  minimality of these linear codes.

Let $k$ be a positive integer and $m=2k$, $2\leq s\leq q^k+1$ be a positive integer and
\begin{equation*}
\Omega=\{E_i\leq \mathbb{F}_q^m:\ {\rm{dim}}E_i=k,E_i\cap E_j=\{\mathbf{0}\},1\leq i\neq j\leq s\}
\end{equation*}
 be a partial spread of $\mathbb{F}_q^{m}$. Let
 \begin{equation}\label{eq2}
D=(\bigcup\limits^s\limits_{i=1}E_i)\backslash \{\mathbf{0}\}=\bigcup\limits^s\limits_{i=1}(E_i\backslash \{\mathbf{0}\}).
\end{equation}
It is easy that $\mathcal{C}(D)$ is a $[s(q^{k}-1),m]_q$ linear code.

The following lemma is important in the proofs of this section.
\begin{lem}\label{lem31}\
For all $\mathbf{y}\in \mathbb{F}_q^m\backslash\{\mathbf{0}\}$, $E_i\in \Omega$,   we have $H(\mathbf{y},E_i)=V(\mathbf{y},E_i)$ and
$$
 {\rm{dim}}V(\mathbf{y},E_i)=\left\{
            \begin{aligned}
             &k, & if\quad \mathbf{y}\in {E_i}^\perp;\\
             &k-1, & if\quad \mathbf{y}\notin {E_i}^\perp.\\
            \end{aligned}
          \right.
$$
\end{lem}
\begin{proof}
Since $E_i$ is a subspace, we have $H(\mathbf{y},E_i)=V(\mathbf{y},E_i)$. Since ${\rm{dim}}E_i=k$, by (\ref{dim sum}), we get ${\rm{dim}}{E_i}^\perp=k$. Let $\alpha_1,\cdots,\alpha_k$ be a basis of ${E_i}^\perp$, $A=(\alpha_1,\cdots,\alpha_k,\mathbf{y})^T$ be a $(k+1)\times m$ matrix. Note that $V(\mathbf{y},E_i)$ is the solution space
of $A\mathbf{x}^T=\mathbf{0}$, where $\mathbf{x}=(x_1,\cdots,x_m)\in \mathbb{F}_q^m$. Thus ${\rm{dim}}V(\mathbf{y},E_i)=m-r(A)$. So if $\mathbf{y}\in {E_i}^\perp$, then ${\rm{dim}}V(\mathbf{y},E_i)=k$; if $\mathbf{y}\notin {E_i}^\perp$, then
${\rm{dim}}V(\mathbf{y},E_i)=k-1$. The proof is completed.
\end{proof}
The following lemma is trivial.
\begin{lem}\label{lem32}\
Let $\Omega=\{E_1,...,E_s\}$ be a partial spread of $\mathbb{F}_q^{m}$. Define $\Omega^\perp:=\{E_1^\perp,...,E_s^\perp\}$. Then $\Omega^\perp$ is also
 a partial spread of $\mathbb{F}_q^{m}$. That is to say, for any $1\leq i\neq j \leq s$, ${\rm{dim}}{E_i}^\perp=k$, and $E_i^\perp\cap E_j^\perp=\{\mathbf{0}\}$.
\end{lem}
Now we consider the minimality of $\mathcal{C}(D)$ in three cases. First, when $s\geq q+1$, we have

\begin{thm}\label{thm31}
Let  $\Omega=\{E_1,...,E_s\}$ be a partial spread of $\mathbb{F}_q^{m}$. If $s\geq q+1$, then $\mathcal{C}(D)$ is a $[s(q^{k}-1),m]_q$ minimal linear code.
\end{thm}
\begin{proof}
According to Lemma \ref{lem23}, we only need to prove that for any $\mathbf{y}\in \mathbb{F}_q^m\backslash\{\mathbf{0}\}$, dim$V(\mathbf{y},D)=m-1$. By (\ref{eq2}), we get
\begin{equation}\label{eq5}
H(\mathbf{y},D)=D\cap H(\mathbf{y})=\bigcup\limits^s\limits_{i=1}(H(\mathbf{y},E_i)\backslash \{\mathbf{0}\}).
\end{equation}

When $k=1$, then  $m=2$. For any $1\leq i\neq j\leq s$, $E_i$ is a one dimensional subspace of
$\mathbb{F}_q^2$ and $E_i\cap E_j=\{\mathbf{0}\}$. By Lemma \ref{lem32}, we have that ${E_i}^\perp$ is also a one dimensional subspace of
$\mathbb{F}_q^2$ and ${E_i}^\perp \cap {E_j}^\perp=\{\mathbf{0}\}$. Note that $\mathbb{F}_q^2$ has $q+1$ one dimensional subspaces in all and $s\geq q+1$,
so $s=q+1$ and $\mathbb{F}_q^2=\bigcup\limits^s\limits_{i=1}{E_i}^\perp$. Then for any $\mathbf{y}\in \mathbb{F}_q^2\backslash\{\mathbf{0}\}$, there exists
$i_0$, such that $\mathbf{y}\in {E_{i_0}}^\perp$, so ${\rm{dim}}V(\mathbf{y},E_{i_0})=1$, thus ${\rm{dim}}V(\mathbf{y},D)=1=m-1$.

When $k>1$, there are two cases.

(1) If there exists $i_0$ such that $\mathbf{y}\in E_{i_0}^\perp$, then by Lemma {\ref{lem31}}, ${\rm{dim}}H(\mathbf{y},E_{i_0})=k$. According to
Lemma {\ref{lem31}} and Lemma {\ref{lem32}}, we can take $j_0\neq i_0$
such that ${\rm{dim}}H(\mathbf{y},E_{j_0})=k-1$. Let $\alpha_1,\cdots,\alpha_k$ be a basis of $H(\mathbf{y},E_{i_0})$ and
$\beta_1,\cdots,\beta_{k-1}$ be a basis of $H(\mathbf{y},E_{j_0})$.
By (\ref{eq5}), we have
\begin{equation}\label{eq6}
H(\mathbf{y},D)\supseteq \{\alpha_1,\cdots,\alpha_k,\beta_1,\cdots,\beta_{k-1}\}.
\end{equation}
Since $E_{i_0}\cap E_{j_0}=\{\mathbf{0}\}$,
\begin{equation}\label{eq7}
r(\{\alpha_1,\cdots,\alpha_k,\beta_1,\cdots,\beta_{k-1}\})=2k-1=m-1.
\end{equation}
Combining (\ref{eq6}) and (\ref{eq7}), we get ${\rm{dim}}V(\mathbf{y},D)=m-1$.

(2) If for any $1\leq i\leq s$, $\mathbf{y}\notin E_i^\perp$, then ${\rm{dim}}H(\mathbf{y},E_i)=k-1$. Let $\alpha_1,\cdots,\alpha_{k-1},\alpha_k$ be a basis of $E_1$, where $\alpha_1,\cdots,\alpha_{k-1}$ is a basis of $H(\mathbf{y},E_1)$. Let $\beta_1,\cdots,\beta_{k-1},\beta_k$ be a basis of $E_2$,
where $\beta_1,\cdots,\beta_{k-1}$ be a basis of $H(\mathbf{y},E_2)$. Let
$$B=\{\alpha_1,\cdots,\alpha_{k-1},\beta_1,\cdots,\beta_{k-1}\}.$$
Since $E_1\cap E_2=\{\mathbf{0}\}$, we have $r(B)=2k-2=m-2$.

Let $V=\mathbb{F}_q^m$, $W=\mathrm{Span}(B)$ and $\overline{V}=V/W$ the quotient space. Then ${\rm{dim}}_{\mathbb{F}_q}\overline{V}=2$.
Let $\pi$ be the standard mapping from $V$ to $\overline{V}$. Note that $\pi$ is a linear mapping, then for any $1\leq i\leq s$,
$\pi(E_i)$ is a subspace of $\overline{V}$. If $\pi(E_i)=\{\overline{\mathbf{0}}\}$, then we get $E_i\leq W\leq H(\mathbf{y})$,
thus $\mathbf{y}\in {E_i}^\perp$, a contradiction. So ${\rm{dim}}\pi(E_i)=1$ or 2.

(i) If there exists $i_0$ such that ${\rm{dim}}\pi(E_{i_0})=2$, then $\pi(E_{i_0})=\overline{V}$. Let $b=(\alpha_k\mathbf{y}^T)/(\beta_k\mathbf{y}^T)$.
Then there exists $\alpha\in E_{i_0}$ such that $\pi(\alpha)=\overline{\alpha_k-b\beta_k}$. So $\alpha=\alpha_k-b\beta_k+\mathbf{w}$, $\mathbf{w}\in W$. It is easy to see that $\alpha \neq \mathbf{0}$, $\alpha \notin W$ and $\alpha \in H(\mathbf{y})$.
By (\ref{eq5}), we have
\begin{equation}\label{eq8}
H(\mathbf{y},D)\supseteq (B\cup \{\alpha\}).
\end{equation}
while
\begin{equation}\label{eq9}
r(B\cup \{\alpha\})=m-1.
\end{equation}
Combining (\ref{eq8}) and (\ref{eq9}), we have ${\rm{dim}}V(\mathbf{y},D)=m-1$.

(ii) If for any $1\leq i\leq s$, ${\rm{dim}}\pi(E_i)=1$. Since for any  $1\leq i\neq j\leq s$ , we have $V=E_i+E_j$. So we get $\overline{V}=\pi(V)=\pi(E_i+E_j)=\pi(E_i)+\pi(E_j)$. Thus for any  $1\leq i\neq j\leq s$, $\pi(E_i)\neq \pi(E_j)$. Since $\overline{V}$ has $q+1$ one dimensional subspaces in all and $s\geq q+1$. Hence, in this case, $s=q+1$ and $\overline{V}=\bigcup\limits^s\limits_{i=1}\pi(E_i)$. Then there exists $j_0$ such that $\pi(E_{j_0})=\mathrm{Span}(\{\overline{\alpha_k-b\beta_k}\})$. Thus there exists $\alpha\in E_{j_0}$ such that $\pi(\alpha)=\overline{\alpha_k-b\beta_k}$.
Then we have $\alpha=\alpha_k-b\beta_k+\mathbf{w}$, $\mathbf{w}\in W$. It is easy to see that $\alpha \in H(\mathbf{y})$ and $\alpha \notin W$. So
$$r\{B\cup \{\alpha\}\}=m-1, H(\mathbf{y},D)\supseteq B\cup \{\alpha\}.$$
and ${\rm{dim}}V(\mathbf{y},D)=m-1$.

In conclusion, when $k>1$,  for any $\mathbf{y}\in \mathbb{F}_q^m\backslash \{\mathbf{0}\}$, ${\rm{dim}}V(\mathbf{y},D)=m-1$, thus $\mathcal{C}(D)$ is minimal.
\end{proof}
\begin{remark}
Theorem \ref{thm31} is a special case of \cite[Theorem 10]{LYT2022} when $t_0=0$. Our method is different from theirs. When $s\leq q$, our method also can be used to study the minimality of the linear codes, while theirs can not.
\end{remark}

Second, when $s=3$, we have
\begin{thm}\label{thm32}
Let $\Omega=\{E_1,...,E_s\}$ be a partial spread of $\mathbb{F}_q^{m}$. If $s=3\leq q$, then $\mathcal{C}(D)$ is not minimal.
\end{thm}
\begin{proof}
Since $\mathbb{F}_q^m=E_1 +E_2$ and for any $1\leq i \neq j \leq 3$ ,$E_i \cap E_j=\{\mathbf{0}\}$, it is easy to see for any $\alpha \in E_1$, there exists unique $\beta \in E_2$ such that $\alpha+\beta \in E_3$. If  we define a map from $E_1$ to $E_2$, $\varphi(\alpha)=\beta$, then $\varphi $ is a linear isomorphism from $E_1$ to $E_2$ and $E_3=\{\mathbf{x}+\varphi(\mathbf{x})|\mathbf{x}\in E_1\}$.

By Lemma \ref{lem32}, we can get ${\rm{dim}}E_1^\perp={\rm{dim}}E_2^\perp=k$ and $E_1^\perp\cap E_2^\perp=\{\mathbf{0}\}$. Thus for any $\mathbf{y}_1\in E_2^\perp\backslash\{0\}$, we have $\mathbf{y}_1\notin E_1^\perp$ and then ${\rm{dim}}H(\mathbf{y}_1,E_1)=k-1$ by
Lemma \ref{lem31}. Let $\alpha_1,\cdots,\alpha_k$ be a basis of $E_1$, such that
$$
\alpha_i\mathbf{y}_1^T=\left\{
            \begin{array}{ll}
             0, & \hbox{for\ $1\leq i\leq k-1;$}\\
             1, & \hbox{for\ $i=k.$}\\
            \end{array}
          \right.
$$
Since $\varphi$ is a linear isomorphism, we have ${\rm{dim}}\varphi(H(\mathbf{y}_1,E_1))=k-1$ and ${\rm{dim}}\varphi(H(\mathbf{y}_1,E_1))^\perp=k+1$.
Thus
$$
\begin{aligned}
& {\rm{dim}}(\varphi(H(\mathbf{y}_1,E_1))^\perp \cap  E_1^\perp)\\
&={\rm{dim}}(\varphi(H(\mathbf{y}_1,E_1))^\perp)+{\rm{dim}}(E_1^\perp)-{\rm{dim}}(\varphi(H(\mathbf{y}_1,E_1))^\perp+E_1^\perp)\\
& \geq k+1+k-m=1.
\end{aligned}
$$
Since $q\geq3$,  there must exist $\mathbf{y}_2\in (\varphi(H(\mathbf{y}_1,E_1))^\perp \cap E_1^\perp) \backslash \{\mathbf{0}\}$, such that $\varphi(\alpha_k)\mathbf{y}_2^T\neq -1$. Since $E_i^\perp \cap E_j^\perp=\{\mathbf{0}\}$, we have $\mathbf{y}_2\notin E_2^\perp$. Now we take $\mathbf{y}_0=\mathbf{y}_1+\mathbf{y}_2$. Since $\mathbf{y}_1\notin E_1^\perp$
 and $\mathbf{y}_2\in E_1^\perp$, we have $\mathbf{y}_0\notin E_1^\perp$. Since $\mathbf{y}_1\in E_2^\perp$ and $\mathbf{y}_2\notin E_2^\perp$, we have $\mathbf{y}_0\notin E_2^\perp$. Since $\alpha_k+\varphi(\alpha_k)\in E_3$ and
$$
\begin{aligned}
(\alpha_k+\varphi(\alpha_k))\mathbf{y}_0^T &=(\alpha_k+\varphi(\alpha_k))(\mathbf{y}_1+\mathbf{y}_2)^T\\
&=\alpha_k\mathbf{y}_1^T+\alpha_k\mathbf{y}_2^T+\varphi(\alpha_k)\mathbf{y}_1^T+\varphi(\alpha_k)\mathbf{y}_2^T\\
&=1+0+0+\varphi(\alpha_k)\mathbf{y}_2^T\neq0.
\end{aligned}
$$
Hence we have $\mathbf{y}_0\notin E_3^\perp$. So for all $1\leq i \leq 3$, we have ${\rm{dim}}H(\mathbf{y}_0,E_i)=k-1$ by Lemma \ref{lem31}. It is easy to verify that
$$H(\mathbf{y}_0,E_1)=H(\mathbf{y}_1,E_1),$$
$$H(\mathbf{y}_0,E_2)=H(\mathbf{y}_2,E_2)=\varphi(H(\mathbf{y}_1,E_1))=\varphi(H(\mathbf{y}_0,E_1)),$$
$$H(\mathbf{y}_0,E_3)=\{\mathbf{x}+\varphi(\mathbf{x})|\mathbf{x}\in H(\mathbf{y}_0,E_1)\}\subseteq \mathrm{Span}(H(\mathbf{y}_0,E_1)\cup H(\mathbf{y}_0,E_2)).$$
Then ${\rm{dim}}V(\mathbf{y}_0,D)=k-1+k-1=m-2$. By Lemma \ref{lem23}, we have $c(\mathbf{y}_0)$ is not minimal.
\end{proof}
Combining Theorem \ref{thm32} and Corollary \ref{cor21}, we have
\begin{cor}\label{cor31}
Let $\Omega=\{E_1,...,E_s\}$ be a partial spread of $\mathbb{F}_q^{m}$. If $2\leq s\leq3\leq q$, then $\mathcal{C}(D)$ is not minimal.
\end{cor}
Third, when $3<s\leq q$, the situation is somewhat complicated, because for some partial spreads $\Omega=\{E_1,...,E_s\}$, $\mathcal{C}(D)$ are minimal while for some other partial  spreads $\Omega$, $\mathcal{C}(D)$ are not minimal.

First,we give examples such that $\mathcal{C}(D)$ are not minimal. Let $\mathbf{e}_1,\cdots,\mathbf{e}_m$ be the standard basis of $\mathbb{F}_q^m$.
For any $b\in \mathbb{F}_q$, we define
\begin{equation}\label{eq10}
E_b={\rm{Span}}(\{\mathbf{e}_i+b\mathbf{e}_{k+i}\mid 1\leq i \leq k),
\end{equation}
and
\begin{equation}\label{eq11}
\Omega=\{E_b|b\in \mathbb{F}_q\}.
\end{equation}
It is easy to prove that ${\rm{dim}}E_b=k$, and $E_{a}\cap E_{b}=\{\mathbf{0}\}$, for any $a$, $b\in \mathbb{F}_q$, $a\neq b$. That is to say, $\Omega$ is a partial spread of $\mathbb{F}_q^m$.
\begin{thm}\label{thm33}
For the partial spread $\Omega$ defined in (\ref{eq11}), the linear code $\mathcal{C}(D)$ is not minimal.
\end{thm}
\begin{proof}
Let $\mathbf{y}_0=\mathbf{e}_1$. Then for any $b\in \mathbb{F}_q$, we get
$$
\begin{aligned}
H(\mathbf{y}_0,E_b)&={\rm{Span}}(\{ \mathbf{e}_2+b\mathbf{e}_{k+2},\cdots , \mathbf{e}_k+b\mathbf{e}_m\})\\
& \subseteq {\rm{Span}}(\{ \mathbf{e}_2,\cdots, \mathbf{e}_k,\mathbf{e}_{k+2},\cdots,\mathbf{e}_m\}).
\end{aligned}
$$
By (\ref{eq5}), we have
$$H(\mathbf{y}_0,D)\subseteq \mathrm{Span}(\{ \mathbf{e}_2,\cdots, \mathbf{e}_k,\mathbf{e}_{k+2},\cdots,\mathbf{e}_m\}).$$
and then ${\rm{dim}}V(\mathbf{y}_0,D)\leq m-2$. By Lemma \ref{lem23}, we have $\mathbf{c}$$(\mathbf{y}_0)$ is not minimal and $\mathcal{C}(D)$ is not minimal.
\end{proof}
Combining Theorem \ref{thm33} and Corollary \ref{cor21}, we have
\begin{cor}
Let $3<s\leq q$ and $S\subseteq \mathbb{F}_q$ where $\# S=s$. Let $\Omega=\{E_b|\ b\in S\}$. Then $\mathcal{C}(D)$ is not minimal.
\end{cor}
Now, we give examples such that $\mathcal{C}(D)$ are minimal.

Let $f(x)$ be an irreducible polynomial in $\mathbb{F}_q[x] $ of degree $k$ and $M\in M_{k\times k}(\mathbb{F}_q)$ satisfies that the characteristic polynomial of $M$ is $f(x)$. We define
\begin{equation}\label{eq12}
\begin{aligned}
&E_1=\{ (\mathbf{x},\mathbf{0})|\mathbf{x}\in \mathbb{F}_q^k \},E_2=\{ (\mathbf{0},\mathbf{x})|\mathbf{x}\in \mathbb{F}_q^k \},\\
&E_3=\{ (\mathbf{x},\mathbf{x})|\mathbf{x}\in \mathbb{F}_q^k \},E_4=\{ (\mathbf{x},\mathbf{x}M)|\mathbf{x}\in \mathbb{F}_q^k \},
\end{aligned}
\end{equation}
and
\begin{equation}\label{eq13}
\Omega=\{E_1,E_2,E_3, E_4\}.
\end{equation}
It is easy to see ${\rm{dim}}E_i=k$ and $E_i\cap E_j=\{\mathbf{0}\}$ for any $1\leq i\neq j\leq 4$. That is to say, $\Omega$ is a partial spread of $\mathbb{F}_q^m$.
\begin{thm}\label{thm34}
For the partial spread $\Omega$ defined in (\ref{eq13}), the linear code $\mathcal{C}(D)$ is minimal.
\end{thm}
\begin{proof}
According to Lemma \ref{lem23}, we only need to prove that for any $\mathbf{y}\in \mathbb{F}_q^m\backslash\{\mathbf{0}\}$, ${\rm{dim}}V(\mathbf{y},D)=m-1$.
There are two cases.

(1) If exists $i_0$ such that $\mathbf{y}\in {E_{i_0}}^\perp$, then by Lemma \ref{lem31}, ${\rm{dim}}H(\mathbf{y},E_{i_0})=k$. According to Lemma \ref{lem31} and Lemma \ref{lem32}, we can take $j_0\neq i_0$ such that ${\rm{dim}}H(\mathbf{y},E_{j_0})=k-1$. Let $\alpha_1,\cdots,\alpha_k$ is a basis of $H(\mathbf{y},E_{i_0})$ and $\beta_1,\cdots,\beta_{k-1}$ is a basis of $H(\mathbf{y},E_{j_0})$. Note that
\begin{equation}\label{eq14}
H(\mathbf{y},D)\supseteq \{\alpha_1,\cdots,\alpha_k,\beta_1,\cdots,\beta_{k-1}\}.
\end{equation}
Since $E_{i_0}\cap E_{j_0}=\{\mathbf{0}\}$, hence
\begin{equation}\label{eq15}
r\{\alpha_1,\cdots,\alpha_k,\beta_1,\cdots,\beta_{k-1}\}=k+k-1=m-1.
\end{equation}
Combining (\ref{eq14}) and (\ref{eq15}), we get ${\rm{dim}}V(\mathbf{y},D)=m-1$.

(2) If for all $1\leq i\leq 4$, $\mathbf{y}\notin {E_i}^\perp$, then $\mathrm{dim}H(\mathbf{y},E_i)=k-1$. Let $\mathbf{y}=(\mathbf{y}_1,\mathbf{y}_2)$ where $\mathbf{y}_1,\mathbf{y}_2\in \mathbb{F}_q^k$. Next we define two linear transformations $\varphi$, $\psi$ from $\mathbb{F}_q^k$ to $\mathbb{F}_q^k$:
\begin{equation}\label{eq16}
\varphi(\mathbf{x})=\mathbf{x}, \psi(\mathbf{x})=\mathbf{x}M.
\end{equation}
Then
\begin{equation}\label{eq17}
E_3=\{(\mathbf{x},\varphi(\mathbf{x}))|\mathbf{x}\in\mathbb{F}_q^k \}, E_4=\{(\mathbf{x},\psi(\mathbf{x}))|\mathbf{x}\in\mathbb{F}_q^k \}.
\end{equation}
Let
\begin{equation}\label{eq18}
\begin{aligned}
S&={\rm{Span}}\{H(\mathbf{y},E_1)\cup H(\mathbf{y},E_2)\}\\
&=\{(\alpha,\beta)|\alpha,\beta\in\mathbb{F}_q^k,\alpha \mathbf{y}_1^T=0,\beta \mathbf{y}_2^T=0 \}\\&=\{(\alpha,\beta)|\alpha\in H(\mathbf{y}_1),\beta\in H(\mathbf{y}_2)\}.
\end{aligned}
\end{equation}
Then dim$S=2k-2=m-2.$ If $H(\mathbf{y},E_3)\subseteq S$ and $H(\mathbf{y},E_4)\subseteq S$, then by (\ref{eq18}), there exists $\alpha_1,\cdots,\alpha_{k-1}\in H(\mathbf{y}_1)$, $\beta_1,\cdots,\beta_{k-1}\in H(\mathbf{y}_2)$ such that $(\alpha_1,\beta_1),\cdots,(\alpha_{k-1},\beta_{k-1})$ is a basis of $H(\mathbf{y},E_3)$. Note that $H(\mathbf{y},E_3)\subseteq E_3$, by (\ref{eq17}), we can get $\varphi(\alpha_i)=\beta_i$. Next we prove the linear independence of $\alpha_1,\cdots,\alpha_{k-1}$. Assume $a_1\alpha_1+\cdots+a_{k-1}\alpha_{k-1}=0$, $a_i\in\mathbb{F}_q$,  we have
 $$
 a_1(\alpha_1,\beta_1)+\cdots+a_{k-1}(\alpha_{k-1},\beta_{k-1})
 =(\sum_{i=1}^{k-1}a_i\alpha_i,\varphi(\sum_{i=1}^{k-1}a_i\alpha_i))=(\mathbf{0},\mathbf{0}).
 $$
Hence $a_i=0$,where $1\leq i\leq k-1$. So $\alpha_1,\cdots,\alpha_{k-1}$ is linearly independent. Since $\beta_i=\varphi(\alpha_i)$,
then $\beta_1,\cdots,\beta_{k-1}$ is also linearly independent. Note that ${\rm{dim}}H(\mathbf{y}_1)=k-1$ and ${\rm{dim}}H(\mathbf{y}_2)=k-1$, then
$$H(\mathbf{y}_1)=\mathrm{Span}(\{\alpha_1,\cdots,\alpha_{k-1}\}),H(\mathbf{y}_2)=\mathrm{Span}(\{\beta_1,\cdots,\beta_{k-1}\}).$$
Thus $ \varphi(H(\mathbf{y}_1))=H(\mathbf{y}_2)$. Similarly, we can get $\psi(H(\mathbf{y}_1))=H(\mathbf{y}_2)$. Then we have
$$\psi(H(\mathbf{y}_1))=H(\mathbf{y}_2)=\varphi(H(\mathbf{y}_1))=H(\mathbf{y}_1).$$
That is to say, $H(\mathbf{y}_1)$ is the $\psi$-$\text{invariant subspace}$ of $\mathbb{F}_q^k$.

Let $\alpha_1,\cdots,\alpha_{k-1},\alpha_k$ is a basis of $\mathbb{F}_q^k$, where $\alpha_1,\cdots,\alpha_{k-1}$ is a basis of $H(\mathbf{y}_1)$. Then the matrix of $\psi$ with respect to this basis is
 $$
 B=
 \left(
 \begin{array}{cc}
 B_{1} & B_{2} \\
 \mathbf{0} & b
 \end{array}
 \right).
 $$
where $B_{1}$ is the matrix of $\psi|H(\mathbf{y}_1)$ with respect to $\alpha_1,\cdots,\alpha_{k-1}$.
Note that $M$ is the matrix of $\psi$ with respect to the standard basis, thus $M$ and $B$ are similar, they have the same characteristic polynomial. So $$f(x)=|xI-B_1|(x-b),$$
a contradiction with the irreducibility of $f(x)$. Hence, $H(y,E_3)\nsubseteq S$ or $H(\mathbf{y},E_4)\nsubseteq S$. It is easy to see that $r(H(\mathbf{y},E_1)\cup H(\mathbf{y},E_2)\cup H(\mathbf{y},E_3))=2k-1$ or $r(H(\mathbf{y},E_1)\cup H(\mathbf{y},E_2)\cup H(\mathbf{y},E_4))=2k-1$.
So ${\rm{dim}}V(\mathbf{y},D)=2k-1=m-1$.

In conclusion, for any $\mathbf{y}\in \mathbb{F}_q^m\backslash \{\mathbf{0}\}$, ${\rm{dim}}V(\mathbf{y},D)=m-1$. By Lemma \ref{lem23}, $\mathcal{C}(D)$ is minimal.
\end{proof}
Combining Theorem \ref{thm34} and Lemma \ref{lem24}, we have
\begin{cor}
Let $s\geq 4$ and $\Omega=\{E_1,\cdots,E_s\}$ be a partial spread of $\mathbb{F}_q^m$. If $\{E_1,E_2,E_3, E_4\}$ are defined as (\ref{eq12}), then $\mathcal{C}(D)$ is  minimal.
\end{cor}

\section {\bf{Concluding remarks}}\label{section Concluding remarks}
In this paper, we use the geometric approach to
study the minimality of linear codes constructed from partial spreads in all cases. In \cite{LYT2022}, they assume that $T_0$ is a $\mathbb{F}_p-$subspace of $\mathbb{F}_p^m$ with dim$_{\mathbb{F}_p}{T_0}$=$t_0$, $l=k+t_0$ and $$\Omega=\{E_i\leq \mathbb{F}_p^m:\ {{\rm{dim}}_{\mathbb{F}_p}}E_i=l,\ E_i\cap E_j=T_0,\ 1\leq i\neq j\leq s\}.$$ In \cite[Theorem 10]{LYT2022}, they prove that when $s=\#\Omega>p$, the corresponding linear code is minimal. Our results in this paper generalize the special case in \cite[Theorem 10]{LYT2022} when $t_0=0$ to all $s\geq 2$. When $t_0\neq 0$, the problem becomes complicated and we will consider it in the future.

 {}
\end{document}